\documentclass[12pt]{article}
\usepackage{amsthm,amsmath, graphicx, booktabs, multirow, xcolor, enumitem}
\usepackage{graphicx,psfrag,epsf}
\RequirePackage[numbers]{natbib}

\newcommand{\blind}{1}
\numberwithin{equation}{section}
\theoremstyle{plain}
\newtheorem{theorem}{Theorem}[section]

\newtheorem{corollary}{Corollary}[section]

\newcommand{\bm}{\boldsymbol}
\newcommand{\bmath}{\boldsymbol}

\begin{document}

\bibliographystyle{agsm}

\def\spacingset#1{\renewcommand{\baselinestretch}%
{#1}\small\normalsize} \spacingset{1}


\if1\blind
{
  \title{\bf A Semiparametric Approach to Model Effect Modification}
  \author{Muxuan Liang\\
    Public Health Sciences Division\\
    Fred Hutchinson Cancer Research Center\\
    Seattle, Washington 98109\\
    and \\
    Menggang Yu\thanks{
    	Research reported in this article was partially funded through a Patient-Centered Outcomes
    	Research Institute (PCORI) Award (ME-1409-21219). The
    	views in this publication are solely the responsibility of the authors and do not necessarily
    	represent the views of the PCORI, its Board of Governors or Methodology Committee.}\hspace{.2cm}\\
    Department of Biostatistics and Medical Informatics\\
    University of Wisconsin-Madison\\}
  \maketitle
} \fi

\if0\blind
{
  \bigskip
  \bigskip
  \bigskip
  \begin{center}
    {\LARGE\bf A Semiparametric Approach to Model Effect Modification}
\end{center}
  \medskip
} \fi

\begin{abstract}
One fundamental statistical question for research areas such as precision medicine and health disparity is about discovering effect modification of treatment or exposure by observed covariates. We propose a semiparametric framework for identifying such effect modification. Instead of using the traditional outcome models, we directly posit semiparametric models on contrasts, or expected differences of the outcome under different treatment choices or exposures. Through semiparametric estimation theory, all valid estimating equations, including the efficient scores, are derived. Besides doubly robust loss functions, our approach also enables dimension reduction in presence of many covariates. The asymptotic and non-asymptotic properties of the proposed methods are explored via a unified statistical and algorithmic analysis. Comparison with existing methods in both simulation and real data analysis demonstrates the superiority of our estimators especially for an efficiency improved version. 
\end{abstract}

\noindent%
\vfill

\newpage
\spacingset{1.45} 
\section{Introduction}
\label{sec:intro}

{In many scientific investigations, estimation of the effect modification is a major goal. For example, in precision medicine research, recommending an appropriate treatment among many existing choices is a central question. Based on patient's characteristics, such recommendation amounts to estimating treatment effect modification \cite{Kraemer2013}. Another example is health disparity research that focuses on measuring modification of the association between disparity categories (e.g. race and socioeconomic status) and health outcomes. The estimated effect modification can be utilized to improve the health system \cite{Braveman2006}.}

{In the classical regression modeling framework, this amounts to estimating interactions between covariates and a certain interested variable. Take the precision medicine example, the goal is to find how the patient characteristics interact with the treatment indicator. If the interest focuses on treatment recommendation, then main effects of these characteristics do not directly contribute to it because they are the same for all treatment choices. Similarly for the health disparity example,  the goal is to find how the modifiers interact with the disparity categories. If the interest focuses on elimination of disparity, then main effects of modifiers are {of less importance} because they are the same for all disparity categories.}

Traditionally effect modification or statistical interaction discovery is conducted mainly by testing or estimating product terms in outcome models. Such discovery is hard as it usually requires large sample sizes \cite{Greenland1993}, especially when many covariates are present.  Recent works in the area of precision medicine illustrate that when the goal is treatment recommendation, investigation on the product term in an outcome model may not be ideal as the outcome is also affected by covariate main effects \cite{Zhao2012, Tian2014, Xu2015, Chen2017, zhang2012class, lu2013}. As we have discussed above, these main effects usually are not directly related to treatment recommendation. {Therefore these works focus on learning contrast functions which are differences of conditional expectations of the outcome under two treatment choices. Nonetheless, there is a lack of the literature on how the main effects or estimation of the main effects can contribute to the efficiency of learning such contrast function.}

Most of the existing works use either nonparametric \cite{Zhao2012, zhang2012class} or parametric approaches \cite{Kraemer2013, lu2013, Xu2015}. The nonparametric approaches are flexible but may not be ideal when faced with a large number of covariates.  {The parametric approaches on the other hand can be sensitive to the underlying model assumptions.}  \citet{song2017} considered a single index model for the contrast function to fill an important middle ground. Single index models are semiparametric models where the index is formed from a linear combination of covariates and a wrapper function that takes the index as argument is nonparametric. 
However only an intuitive method of estimation was considered in \citet{song2017}. No systematic investigation was given to explore other possible estimating equations. Therefore issues such as efficiency were left largely untackled. 

More importantly, it is practical to provide more flexibility in the semiparametric framework by allowing more than one indices. That is multiple index models can better capture the heterogeneity in effect modification. As a simple example, a single index model with the linear index part depending on the product of two covariates is not a single index model any more, if this product is not included as a fitting covariate. However multiple index models can easily capture this deviation from the linearity. When there are more than two treatments, it is also mathematically appealing to consider multiple index models. For example, single index models can be used to model the contrasts between treatments A and B and between B and C respectively. But if the indices of these two models are different, the resulting contrast between A and C will be a double index model, not a single index model. This asymmetry, of assuming two single index models for two constrasts and one double index model for the other contrast, is easily avoided by assuming the multiple index models for all contrasts.

We therefore propose a more general semiparametric approach which is essentially a multiple index modeling framework for multiple treatments. We will also consider determination of the number of indices. Under our framework, we make the following new contributions. First, based on the well-established semiparametric estimation theory \citep{BKRWbook, Tsiatisbook}, we characterize all valid estimating equations, including the efficient score under our framework. This leads to many possible choices of estimating equations, and efficiency consideration becomes very natural in our approach. Second, because multiple index models are intrinsically related to dimension reduction \citep{Xia2002, Xia2007}, our method can also be used as a dimension reduction tool for interaction discovery with a specific variable. Third, we do not restrict the treatment or exposure variable to be binary. Literature for more than two treatment choices seem very sparse \citep{Lou2017}. Fourth, we also study the asymptotic and non-asymptotic properties of the resulting estimators based on a careful analysis of the computing algorithm. This enables inference and provides useful insights for using our approach in practice. 

Estimating the effect modification is an important problem in causal inference \cite{abrevaya2015,kosuke2013}. Under the potential outcome framework \cite{rubin1974, rubin2005}, and the well-known assumptions of the Stable Unit Treatment Value Assumption (SUTVA), consistency, and treatment assignment ignorability \cite{imbens_rubin_2015}, the effect modification becomes the conditional average treatment effect (CATE). Under these assumptions, popular methods such as inverse probability weighting (IPW) and augmented inverse probability weighting (AIPW) \cite{robins1994, robins2005, Tsiatis2009, Tan2010, robins2012} were commonly used to estimate average treatment effect (ATE) \cite{hirano2001, hirano2003} and the CATE \cite{abrevaya2015,kosuke2013}. 
{On estimating the CATE, many literature also chose to directly work with outcome models \cite{Green2012, Xie2012, Lu2018, Wager2018, Kunzel2019}.} The well-known structural nested models and the corresponding G-estimation focused on parametric models for the CATE with relatively few covariates \cite{Robins1992b, Robins1994b, vansteelandt2014}. We posit a multiple index model on the contrast function or the CATE and show how the main effects contribute to the efficiency. Our proposed approach in some way extends these results on the CATE in a semiparametric modeling framework.  In some literature \cite{luo2017,huang2017,Persson2017}, the effect modification appears to be used also as an important middle step to estimate the population level causal quantities such as the ATE. However, the methods proposed in these literature, including index models or dimension reduction, are for the outcomes, not for the contrast functions.

\section{A semiparametric framework for modeling contrast functions}
\label{s:method}


Suppose $\bmath{X}\in \mathcal{X}$ is a $p$-dimensional vector of covariates, $Y$ is an outcome, and $T$ is a discrete variable whose effect on $Y$ and modification of this effect by ${\bm X}$ are of interest. We first consider the case when $T$ has only two levels. We can also use $\{1, 2\}$, instead of $\{-1, 1\}$, to denote the levels of $T$ and to conform with our notation below for the more general case. However we keep  $\{-1, 1\}$ as it leads to simpler notation in our presentation. 

The main goal is to learn the following contrast function based on observed data, 
\begin{equation}\label{eq:model_contrast}
\Delta(\bmath{X}) \equiv E[Y|T=1,\bmath{X}]-E[Y|T=-1,\bmath{X}].
\end{equation}
We assume that a larger $Y$ is better. Then when $\Delta(\bmath{X})>0$, $T=1$ rather than $T=-1$ leads to a better clinical outcome for given $\bmath{X}$, and vice versa. Therefore we consider the following model in this article 
\begin{equation}\label{eq:model_contrast_g}
\Delta(\bmath{X})=g(\bmath{B}_0^\top \bmath{X})
\end{equation}
 where $g$ is an unknown function and $\bmath{B}_0$ is a $p\times d$ matrix. {Here $d$ represents the number of indices. That is, $d=1$ corresponds to a single index model and $d>1$ to a multiple index model. }

 Note that there is an identifiability issue in Model \eqref{eq:model_contrast_g} when both $g$ and $\bmath{B}_0$ are unrestricted. This is a known issue in both the index models and dimension reduction literature \citep{XIA20061162, Xia2008,  BingLi2018, cook2007, Ma2012, Ma2013, Xia2002, Xia2007}.  {To resolve this issue, we adopt the common strategy in  the dimension reduction literature \citep{BingLi2018, cook2007, Ma2012} and assume that the columns of $\bmath{B}_0$ form a Grassmann manifold. That is, $\bmath{B}_0$ satisfies}
 \begin{equation*}
 \begin{pmatrix}
 	{\bm I}_{d\times d}, {\bm 0}_{d\times (p-d)}
 \end{pmatrix}
\bmath{B}_0 = {\bm I}_{d\times d}.
\end{equation*} where ${\bm I}_{d\times d}$ is the identity matrix with rank $d$.

Model \eqref{eq:model_contrast_g} is very flexible as the contrast function is defined in terms of the conditional means of the outcome, instead of its conditional distributions. The model is therefore semiparametric as it leaves the other parts of the distribution (e.g. variance)  unspecified. This is similar to the well known semiparametric conditional mean model commonly used in  econometrics \cite{CHAMBERLAIN1987, Newey2004}.

Consequently, the outcome $Y$ can be of many types as long as its mean function satisfies our model. For example, when $Y$ is binary, the contrast function represents the difference of the success probabilities. Then Model \eqref{eq:model_contrast_g} implies a single or multiple index model, depending on $d=1$ or $d>1$, for the difference of its success probabilities under the two treatment choices.

Now consider the case when $T$ has $K$ levels. To fully represent the effect modification, we need to use $K-1$ contrasts. For example when $K=3$, we can use contrasts such as $ E[Y|T=1,\bmath{X}]-E[Y|T=2,\bmath{X}]$ and $E[Y|T=3,\bmath{X}]-\frac{1}{2}(E[Y|T=1,\bmath{X}]+E[Y|T=2,\bmath{X}])$. In general, we extend the concept of the contrast function in (\ref{eq:model_contrast}) to a contrast vector function of length $K-1$ as follows
\begin{equation}\label{eq:trteffect_multi}
{\bm \Delta}(\bmath{X}) \equiv {\bm\Omega}\left(
\begin{matrix}
E[Y|T=1,\bmath{X}]\\
\vdots\\
E[Y|T=K,\bmath{X}]
\end{matrix}
\right).
\end{equation}
where
${\bm\Omega}$ is a pre-specified $(K-1)\times K$ matrix. The $K-1$ rows of ${\bm\Omega}$ represent the interested contrasts. For $K=2$, ${\bm\Omega}=(1, -1)$.  For the above example of $K=3$, we have
\begin{eqnarray*}
	{\bm\Omega}=
	\begin{pmatrix}
		1 & -1 & 0 \\
		-\frac{1}{2} & -\frac{1}{2} &1
	\end{pmatrix}.
\end{eqnarray*}
For the contrasts to be interpretable, we require the sum of $i$th row of ${\bm\Omega}$ to be 0, that is, $\sum_{j=1}^K{\bm\Omega}_{ij}=0$ for $i=1, \dots, K-1$. Reasonably, we also require ${\bm\Omega}{\bm\Omega}^\top$ to be invertible. 

In this setup, the corresponding model is
\begin{equation}\label{eq:model_contrast_multi}
{\bm \Delta}(\bmath{X}) = {\bm g} (\bmath{B}_0^\top\bmath{X}),
\end{equation} 
where ${\bm g}$ is a length $(K-1)$ vector function of  $\bmath{B}_0^\top\bmath{X}$.

\section{Tangent spaces and semiparametric efficient scores}


Similar to the work in dimension reduction \citep{Ma2012, Ma2013, MA2014}, we characterize the nuisance tangent space and its orthogonal complement for $\bmath{B}_0$. The corresponding efficient score is also derived. We closely follow the notions and techniques of \citet{Tsiatisbook}. The derivation requires working with the full data likelihood even though we do not specify the form of the distribution of $Y$ in Model \eqref{eq:model_contrast_g} or \eqref{eq:model_contrast_multi}. In other words, we need to convert these models into equivalent outcome models that involve $\bmath{B}_0$, $g$, and the unspecified nonparametric parts.  

In our Supplemental Materials, we show that our model for binary treatments  \eqref{eq:model_contrast} is equivalent to the following model for the outcome $Y$:
\begin{equation}\label{eq:model2}
Y=\frac{1}{2}Tg(\bmath{B}_0^\top \bmath{X})+\epsilon
\end{equation}
where $\epsilon$ is some random variable satisfying the following conditional mean condition
\begin{equation}\label{eq:condition_error}
E\left [\epsilon \middle\vert T, \bmath{X}\right ]=E\left [\epsilon \middle\vert  \bmath{X}\right ].
\end{equation}
The equivalence can be shown by verifying that  $\epsilon \equiv Y-\frac{1}{2}Tg(\bmath{B}_0^\top \bmath{X})$  satisfies (\ref{eq:condition_error}). This representation \eqref{eq:model2} enables us to directly work with the full data likelihood. 

Similar to the binary setting, when $T$ is multi-level, our model \eqref{eq:model_contrast_multi}  is equivalent to the following model for the outcome $Y$:
\begin{equation}\label{eq:model2_multi}
Y={\bm\Omega}_{\cdot T}^\top \left ({\bm\Omega}{\bm\Omega}^\top\right)^{-1} {\bm g}(\bmath{B}_0^\top\bmath{X})+\epsilon \end{equation}
where ${\bm\Omega}_{\cdot T}$ is the column of ${\bm\Omega}$ that corresponds to the value of the treatment $T$. 
Similarly $\epsilon$ in \eqref{eq:model2_multi} needs to satisfy the condition \eqref{eq:condition_error}.

We first present results for the general multi-level $T$ and assume that the function class of interest is  the mean zero Hilbert space $\mathcal{H}=\{f(\epsilon, \bmath{X}, T):E(f)=0\}$. These results will then be simplied for binary treatments.


The full data likelihood is 
\begin{equation*}
p_{\bmath{X}}(\bmath{X})\pi_T(\bmath{X})p_{\epsilon}\left (Y-{\bm\Omega}_{\cdot T}^\top \left ({\bm\Omega}{\bm\Omega}^\top\right)^{-1} {\bm g}(\bmath{B}_0^\top\bmath{X}),\bmath{X}, T\right ),
\end{equation*}
where $p_{\bmath{X}}$ is the density of $\bmath{X}$, $\pi_T(\bmath{X})$ is the density of $T$ conditional on $\bm X$, and $p_{\epsilon}$ is the density of $\epsilon$ conditional on $\bm X$ and $T$, with respect to some dominating measure. The density $\pi_T(\bmath{X})$ is  also known as propensity score \citep{Rosenbaum1983}. Note that $p_{\bmath{X}}$, $\pi_T$, $p_{\epsilon}$, and $g$ are infinite-dimensional nuisance parameters.The tangent spaces correspond to $p_{\bmath{X}}$, $p_{\epsilon}$, and $\pi_T$ are 
\begin{eqnarray*}
	\Lambda_{\bmath{X}}&=&\{f(\bmath{X})\in \mathcal{H}: E[f]=0\} \label{eq:Lambda2s_multi}\\%
{\Lambda_{\epsilon}} &=& \left \{f(\epsilon,\bmath{X}, T)\in \mathcal{H}: E(f|\bmath{X}, T)=0 
	\text{ and } E\big[f\epsilon\middle\vert T, \bmath{X}\big]= E\big[f\epsilon\middle\vert\bmath{X}\big]\right \}.\label{eq:Lambda1s_multi} \\%
	\Lambda_{\pi} &=& \{f(\bmath{X}, T)\in \mathcal{H}: E[f\,|\,\bmath{X}]=0\}.
\end{eqnarray*}
Through some algebra, we can rewrite $\Lambda_{\pi}$ as
\begin{equation*}\label{eq:Lambdapi_multi}
\Lambda_{\pi} = \left \{\bmath{w}_T^\top \left ({\bm\Omega}{\bm\Omega}^\top\right)^{-1}\bm f_{\pi}(\bmath{X}), \forall \bm f_{\pi}(\bmath{X}):\mathcal{X}\mapsto R^{K-1}\right \}.
\end{equation*} where  $$\bmath{w}_T=\frac{{\bm\Omega}_{\cdot T}}{\pi_T(\bmath{X})}.$$ 
The tangent space of $\bm g$ is
\begin{equation*}\label{eq:Lambdag_multi}
\Lambda_{\bm g} = \left \{\frac{p^{'}_{\epsilon, 1}(\epsilon,\bmath{X}, T)}{p_{\epsilon}(\epsilon,\bmath{X}, T)}{\bm\Omega}_{\cdot T}^\top \left ({\bm\Omega}{\bm\Omega}^\top\right)^{-1}\bm f_{\bm g}(\bmath{B}_0^\top \bmath{X}),\forall \bm f_{\bm g}(\bmath{B}_0^\top \bmath{X}):\mathcal{X}\mapsto R^{K-1}\right \},
\end{equation*}
where $p^{'}_{\epsilon, 1}(\cdot)$ is the derivative of $p_{\epsilon}(\epsilon, \bmath{X}, T)$ w.r.t $\epsilon$.

Let ${\perp}$ denote the orthogonal complement of a Hilbert space. Denote the nuisance tangent space $\Lambda\equiv \Lambda_{\bmath{X}}+\Lambda_{\epsilon}+\Lambda_{\pi}+\Lambda_{\bm g}$. Then we have
\begin{theorem}\label{thm:nuissance_multi}
	The orthogonal complement of the nuisance tangent space, $\Lambda^{\perp}$,  is a subspace characterized by all functions with the form
	\begin{equation*}
	\bmath{w}_T^\top\left [\epsilon- E(\epsilon|\bmath{X})\right ]\left [\bm{\alpha}(\bmath{X})-E\{\bm \alpha(\bmath{X})|\bmath{B}_0^\top \bmath{X}\}\right ],
	\end{equation*}
	for any function $\bmath{\alpha}(\bmath{X}):\mathcal{X}\mapsto R^{K-1}$. 
\end{theorem}
Detailed proofs of this theorem and other theorems and corollaries are given in the Supplemental Materials. To obtain the efficient score, we need to project the score function onto $\Lambda^{\perp}$. The following theorem provides a formula to project any function onto $\Lambda^{\perp}$ and thus contains the efficient score as a special case. 
\begin{theorem}\label{lemma:projection_multi}
	For any function $f(\epsilon, \bmath{X}, T)\in \mathcal{H}$, its projection onto $\Lambda^{\perp}$ is given by
	\begin{equation*}
	\bmath{w}_T^\top \left \{\epsilon-E(\epsilon|\bmath{X}) \right \}\bm C(\bmath{B}_0^\top\bmath{X}),
	\end{equation*}
	where 
	\begin{eqnarray*}
		{\bm C}(\bmath{B}_0^\top\bmath{X})&=&{\bm V}({\bm X})\big \{\, \bm D(\bmath{X})-E[{\bm V}({\bm X})|\bmath{B}_0^\top \bmath{X}]^{-1}E[{\bm V}({\bm X})\bm D(\bmath{X})|\bmath{B}_0^\top \bmath{X}]\, \big\},\\
		{\bm V}({\bm X})^{-1}
		&=&E(\bmath{w}_T\bmath{w}_T^\top \epsilon^2|\bmath{X})-E(\bmath{w}_T\bmath{w}_T^\top|\bmath{X})E(\epsilon|\bmath{X})^2, \\%
		\bm D(\bmath{X})&=&E(\bmath{w}_T f \epsilon|\bmath{X})-E(\bmath{w}_T f |\bmath{X})E(\epsilon|\bmath{X}).
	\end{eqnarray*}
\end{theorem}

Note that $\bm C(\bmath{B}_0^\top\bmath{X})$ depends on ${\bm X}$, in addition to $\bmath{B}_0^\top\bmath{X}$. But we have suppressed it for notational simplicity. After setting $f$ as the score function in Theorem \ref{lemma:projection_multi}, we obtain the efficient score in the following corollary.

\begin{corollary}\label{thm:effcient_score_multi}
	The efficient score of $\bmath{B}$ is given by the vectorization of a $d\times p$ matrix whose $(i,j)$ coordinate is given by
	$$\bmath{w}_T^\top \left \{\epsilon-E(\epsilon|\bmath{X}) \right \}\bm C_{i,j}^*(\bmath{B}_0^\top\bmath{X}),$$ where
	\begin{align*}
	\bm C_{i,j}^*(\bmath{B}_0^\top\bmath{X})=&{\bm V}({\bm X})\left \{X_{j}-E[{\bm V}({\bm X})|\bmath{B}_0^\top \bmath{X}]^{-1}E[{\bm V}({\bm X})X_{j}\, |\,\bmath{B}_0^\top \bmath{X}]
	\right \} \times
	\partial_i{\bm g} (\bmath{B}_0^\top \bmath{X}),
	\end{align*}
	$X_{j}$ is the $j$th component of ${\bm X}$, and $\partial_i{\bm g} $ is the derivative of ${\bm g}$ with respect to its $i$th index.
\end{corollary}

In cases like clinical trials, $\pi_T({\bm X})$ may be {\em known}. In this case, there is no corresponding tangent space $\Lambda_{\pi}$ and the corresponding nuisance tangent space
$\tilde{\Lambda}\equiv \Lambda_{\bmath{X}}+\Lambda_{\epsilon}+\Lambda_{\bm g}$. Its
orthogonal complement $\tilde{\Lambda}^{\perp}$ is then larger and can be shown to be the sum of  $\Lambda^{\perp}$ and $\mathcal{S}_2$ defined in the Supplemental Materials. For any function $f(\epsilon, \bmath{X}, T)$, its projection on $\tilde{\Lambda}^{\perp}$ is its projection on ${\Lambda}^{\perp}$ plus an additional term $\bmath{w}_T^\top E(\bmath{w}_T\bmath{w}_T^\top|\bmath{X})^{-1}E(\bmath{w}_T f |\bmath{X})$. However the efficient score is unchanged as $E(\bmath{w}_T f |\bmath{X})=0$ when $f$ is chosen as the score function.

As a special case of Theorem \ref{lemma:projection_multi} and Corollary \ref{thm:effcient_score_multi}, when $K=2$,  we have the following corollaries, recognizing that $\bmath{w}_T =\pi_T(\bmath{X})^{-1} T$ now becomes a scalar.
\begin{corollary}\label{thm:ortho_space}
	For $K=2$ and $T\in \{-1, 1\}$, 
	\begin{equation*}
	\Lambda^{\perp}=\big \{ \, \pi_T({\bm X})^{-1} T \big[\alpha(\bmath{X})-E\{\alpha(\bmath{X})|\bmath{B}_0^\top \bmath{X}\} \big ] \left [\epsilon-E(\epsilon|\bmath{X})\right ], \forall \alpha(\bmath{X}):\mathcal{X}\mapsto R \, \big \}.
	\end{equation*}  
\end{corollary}

\begin{corollary}\label{lemma:projection}
	For $K=2$ and $T\in \{-1, 1\}$, the projection of any function $f(\epsilon, \bmath{X}, T)\in \mathcal{H} $ onto $\Lambda^{\perp}$ is given by
	$$
	\pi_T({\bm X})^{-1} T\, C(\bmath{B}_0^\top\bmath{X})\left \{\epsilon-E[\epsilon|\bmath{X}]\right \},
	$$
	where
	\begin{eqnarray*}
		C(\bmath{B}_0^\top\bmath{X})&=&V({\bm X})\left \{D(\bmath{X})-\frac{E[V({\bm X})D(\bmath{X})|\bmath{B}_0^\top \bmath{X}]}{E[V({\bm X})|\bmath{B}_0^\top \bmath{X}]}\right \} \\
		V(\bmath{X})^{-1} &=& E[\pi_T({\bm X})^{-2}\epsilon^2|\bmath{X}]-E[\pi_T({\bm X})^{-2}|\bmath{X}]E(\epsilon|\bmath{X})^2  \\
		D(\bmath{X})&=&E[	\pi_T({\bm X})^{-1} T f\epsilon|\bmath{X} ] - E[	\pi_T({\bm X})^{-1} T f|\bmath{X} ] E(\epsilon|\bmath{X}).
	\end{eqnarray*}
	Therefore, the efficient score is
	\begin{equation*}
		\pi_T({\bm X})^{-1} T\, \bm C^*(\bmath{B}_0^\top\bmath{X})\left \{\epsilon-E(\epsilon|\bmath{X})\right \},
	\end{equation*}
	where
	\begin{equation*}
	\bm C^*(\bmath{B}_0^\top\bmath{X})=V({\bm X})\, \nabla g(\bmath{B}_0^\top \bmath{X})\otimes \left \{\bmath{X}-\frac{E[V({\bm X})\bmath{X}|\bmath{B}_0^\top \bmath{X}]}{E[V({\bm X})|\bmath{B}_0^\top \bmath{X}]}\right \},
	\end{equation*}
	and $\otimes$ is Kronecker product. 
\end{corollary}

\section{Estimation and algorithm}

We first consider estimation of $\bmath{B}_0$ with fixed $d$. Then we propose a method for determining $d$ similar to \citet{Xia2002}. For simplicity, we present our method with $K=2$. Generalization to $K>2$ is straightforward and relegated to the Supplemental Materials. From Corollary \ref{lemma:projection}, the efficiency score can be written as
	\begin{equation}
V({\bm X}) \frac{T}{\pi_T({\bm X})} \,  \nabla g(\bmath{B}_0^\top \bmath{X})\otimes \left \{\bmath{X}-\frac{E[V({\bm X})\bmath{X}|\bmath{B}_0^\top \bmath{X}]}{E[V({\bm X})|\bmath{B}_0^\top \bmath{X}]}\right \}\left \{\epsilon-E(\epsilon|\bmath{X})\right \}. \label{efficiency score form}
\end{equation}
We can see that the efficient score is hard to estimate directly due to many conditional expectations involved. We therefore use \eqref{efficiency score form}  to accomplish two tasks. 


The first task is to construct more practical and simplified estimation procedures by exploring the robustness of the efficient score \eqref{efficiency score form}.  In particular, \eqref{efficiency score form}  remains unbiased (for 0) by omitting the fraction ${E[V({\bm X})\bmath{X}|\bmath{B}_0^\top \bmath{X}]}/{E[V({\bm X})|\bmath{B}_0^\top \bmath{X}]}$ and the leading term $V({\bm X})$.  In addition, $\pi_T({\bm X})$ and  $E[\epsilon | \bmath X]$ form a pair for robustness in the sense that, if one is known or consistently estimated, the other can be mis-specified. This is the well-known double robustness property in semiparametric estimation \citep{Tsiatisbook}. Therefore we propose the following class of estimating equations are all unbiased for estimating $\bmath{B}_0$ under Model (\ref{eq:model2_multi}), \begin{equation*}
\tilde{S}=\big\{	\pi_T({\bm X})^{-1} T \nabla g(\bmath{B}_0^\top \bmath{X})\otimes \bmath{X} (\epsilon-\eta(\bmath{X})), \forall \eta(\bmath{X}):\mathcal{X}\mapsto R\big\}.
\end{equation*}
This will be our choice of estimating equations. The obvious benefit of using this function class $\tilde{S}$ is that solving the estimating equations is equivalent to minimizing the loss function $\pi_T({\bm X})^{-1}{\{Y-\frac{1}{2}Tg(\bmath{B}_0^\top \bmath{X})-\eta(\bmath{X})\}^2}$. The corresponding sample version is 
\begin{equation}\label{eq:loss}
L_g(\bmath{B})=\frac{1}{n}\sum_{i=1}^n\frac{\{Y_i-\frac{1}{2}T_ig(\bmath{B}^{T} \bmath{X}_i)-\eta(\bmath{X}_i)\}^2}{\pi_{T_i}(\bmath{X}_i)}.
\end{equation} 


The proposed loss function remains doubly robust in the sense that the minimizer of the proposed loss function is consistent if either $\pi_T(\bmath X)$ or $\eta(\bmath X)=E[\epsilon|\bmath X]$ is correctly specified. When $\pi_T(\bmath{X})$ is known or can be consistently estimated, the choice of $\eta(\bmath X)$ can be flexible. A convenient choice is $\eta(\bmath{X})=0$ adopted in \citet{Chen2017} and \citet{Tian2014}. Another choice is $\eta(\bmath{X})=\{1-2\pi(\bmath{X})\}g(\bmath{B}_0^\top \bmath{X})$ used by \citet{song2017}. {However, from the proof of Theorem \ref{thm:nuissance_multi} and Corollary \ref{thm:effcient_score_multi}}, $$\eta^*(\bmath{X})={E[\epsilon|\bmath{X}]}$$ leads to the most efficient estimator. 


Because $g$ is unknown, to estimate $\bmath{B}_0$ through minimizing $L_g(\bmath{B})$, we employ a  minimum average variance estimation (MAVE) type of method as advocated in \citet{Xia2002}. In particular minimization is based on the following approximating loss function:
\begin{align}\label{eq:tMAVE}
&L( \bmath{B}, \{a_j,  \bmath{b}_j\}_{j=1}^n) =\sum_{j=1}^n\sum_{i=1}^n\frac{\{Y_i-\frac{1}{2}T_i[a_j+ \bmath{b}_j^\top ( \bmath{B}^\top  \bmath{X}_i- \bmath{B}^\top  \bmath{X}_j)]-\eta(\bmath{X}_i)\}^2}{n^2 \, \pi_{T_i}(\bmath{X}_i)}w_{ij},
\end{align}
where $w_{ij}=K_h( \bmath{B}^\top  \bmath{X}_j- \bmath{B}^\top  \bmath{X}_i)$ and $K_h(\cdot)=\frac{1}{h^d}K(\cdot/h)$ is a kernel function with bandwidth $h$. The extra parameters $a_j \in R$ and ${\bm b}_j \in R^d$ can be thought of as approximations to $g$ and its gradient at each point $\bmath{B}^\top  \bmath{X}_j$, and the kernel weight $w_{ij}$ ensures the adequacy of the local linear approximation of $g$ in its neighborhood. We can also normalize the weight $w_{ij}$'s by $\tilde{w}_{ij}=w_{ij}/\sum_j w_{ij}$. 
In the next two subsections, we will consider both the case of fixing $\eta(\bmath{X})$ through a sensible or convenient choice and of estimating $\eta^*(\bmath{X})={E[\epsilon|\bmath{X}]}$. We term the two methods interaction MAVE (iMAVE) and iMAVE2 respectively.

	
{The second task is to use the variance of the efficient score  \eqref{efficiency score form}, or the efficiency bound, to evaluate our method. Obviously, our simplified method will lead to efficiency loss in general cases. However, if we further impose two assumptions
	\begin{enumerate}[label=(\alph*)]
		\item \label{cond:1} $\epsilon \perp T | {\bm X}$,  $Var(\epsilon|{\bm X})$ is a constant;
		\item \label{cond:2} $\pi_1(\bm X) \equiv \pi_1$, where $\pi_1$ is a constant,
	\end{enumerate}
Then the efficiency bound (based on the asymptotic variance of the efficient score)  is exactly the same as the variance of our iMAVE2 method derived in Theorem \ref{thm3} below. Therefore, iMAVE2 attains local efficiency under the above two assumptions.}


\subsection{The iMAVE method with a fixed $\eta(\bmath X)$}

In this section, an alternatively weighted least square algorithm to minimize (\ref{eq:tMAVE}) is introduced that consists of the following steps.
\begin{enumerate}
	\item An initial estimator, $\bmath{B}_{(1)}$, is obtained. Please see our comments after the algorithm on how to obtain $\bmath{B}_{(1)}$. 
	\item Let $\bmath{B}_{(t)}$ be the estimator at the $t$th iteration. Calculate $$w_{ij}^{(t)}=K_h( \bmath{B}_{(t)}^\top  \bmath{X}_i- \bmath{B}_{(t)}^\top  \bmath{X}_j).$$
	\item Solve the following weighted least square problem to obtain $$(a_j^{(t)},  \bmath{b}_j^{(t)})=\arg \min_{a_j,  \bmath{b}_j}L_1(a_j,  \bmath{b}_j),$$ for $j=1,\cdots, n$, where
	\begin{equation*}
	L_1(a_j,  \bmath{b}_j)=\frac{1}{n}\sum_{i=1}^n\frac{\{Y_i-\eta(\bmath{X}_i)-\frac{1}{2}T_i[a_j+ \bmath{b}_j^\top ( \bmath{B}_{(t)}^\top  \bmath{X}_i- \bmath{B}_{(t)}^\top  \bmath{X}_j)]\}^2}{\pi_{T_i}(\bmath{X}_i)}w_{ij}^{(t)}.
	\end{equation*}
	
	\item Solve the following weighted least square problem to obtain $$\tilde{\bmath{B}}_{(t+1)}=\arg\min_{\bmath{B}} L_2(\bmath{B}),$$ where
	\begin{align*}
	&L_2(\bmath{B})\\ \nonumber
	=&\frac{1}{n^2}\sum_{j=1}^n\sum_{i=1}^n\frac{\{Y_i-\eta(\bmath{X}_i)-\frac{1}{2}T_i[a_j^{(t)}+ {\bmath{b}_j^{(t)}}^\top ( \bmath{B}^\top  \bmath{X}_i- \bmath{B}^\top  \bmath{X}_j)]\}^2}{\pi_{T_i}(\bmath{X}_i)}w_{ij}^{(t)}.
	\end{align*}
	
	\item {Normalize to obtain $\bmath{B}_{(t+1)}$ by projecting $\tilde{\bmath{B}}_{(t+1)}$ onto the Grassmann manifold.}
	
	\item If the discrepancy, $|\bmath{B}_{(t+1)}-\bmath{B}_{(t)}|$, is smaller than a pre-specified tolerance, or a max number of iterations achieved, then output $\bmath{B}_{(t+1)}$. If not, go back to Step (2) and start a new iteration.
\end{enumerate}

The initial estimator $\bmath{B}_{(1)}$ needs to be a consistent estimator for our theoretical analysis. To get a consistent $\bmath{B}_{(1)}$, one choice is to solve a simplified version of \eqref{eq:tMAVE} by only expanding $g$ at $\bm 0$,
\begin{equation*}
L(\bmath{B})=\frac{1}{n}\sum_{i=1}^n\frac{\{Y_i-\frac{1}{2}T_i \bmath{B}^\top  \bmath{X}_i\}^2}{\pi_{T_i}(\bmath{X}_i)}\tilde{w}_{i0},
\end{equation*}
where
$\tilde{w}_{i0}=K_h(\bmath{B}^\top  \bmath{X}_i)$. For $d=1$, one can also utilize the method of \citet{song2017}. {In our simulation studies, we found that a simple choice of $\bmath{B}_{(1)}=\bmath{0}$ almost always led to stable convergent results.}


\subsection{The iMAVE2 method with an estimated $\eta^*(\bmath X)$}
\label{subsect:eff aug}
The following two-step procedure is proposed to estimate $\eta^*({\bm X})={E[\epsilon|\bmath{X}]}$. First, we obtain an estimate $\hat{\bm B}$ of $\bmath{B}_0$ with a pre-fixed $\eta$. Then $g(\bm B^\top\bm X)$ is estimated by 
\begin{equation}\label{eq:predict_g}
\hat{g}(\hat{\bm B}^\top\bmath{X})=\frac{\sum_{i=1}^{n} \pi_{T_i}(\bmath{X}_i)^{-1} {T_i}Y_iK_{h}(\hat{ \bmath{B}}^\top ( \bmath{X}_i- \bmath{X}))}{\sum_{i=1}^{n} K_{h}(\hat{ \bmath{B}}^\top ( \bmath{X}_i- \bmath{X}))},
\end{equation}
where $K_{h}$ is a kernel function with $K_{h}(\bmath{X})=h^{-d}K(\bmath{X}/h)$. The kernel $K$ and bandwidth $h$ can be different from those used before in \eqref{eq:tMAVE}. 

The estimated residual is $\hat{\epsilon}_i=Y_i-\frac{1}{2}T_i\hat{g}(\hat{\bmath{B}}^\top \bm X_i)$. We can then estimate 
$E[\epsilon|\bmath{X}]$, by
\begin{equation}\label{eq:estimation_mainEffect2}
\frac{\sum_{i=1}^n\hat{\epsilon}_i{K}_{h}(\bmath{X}_i-\bmath{X})}{\sum_{i=1}^n{K}_{h}(\bmath{X}_i-\bmath{X})},
\end{equation}
where $K_{h}$ is another kernel function with $K_{h}(\bmath{X})=h^{-p}K(\bmath{X}/h)$. Again, the kernel $K$ and bandwidth $h$ can be different from those used before.
On the other hand, noticing that $E[\pi_{T_i}(\bmath{X}_i)^{-2}|\bmath{X}]^{-1}= \pi_1(\bmath{X})\pi_{-1}(\bmath{X})$, $\eta^*$ can also be estimated by
\begin{equation}\label{eq:estimation_mainEffect}
\hat{\eta}^*({\bm X})=\pi_1 (\bmath{X})\pi_{-1}(\bmath{X})\frac{\sum_{i=1}^n \pi_{T_i}(\bmath{X}_i)^{-2}\hat{\epsilon}_i{K}_{h}(\bmath{X}_i-\bmath{X})}{\sum_{i=1}^n{K}_{h}(\bmath{X}_i-\bmath{X})}
\end{equation} 
With an estimated $\hat{\eta}^*$, a possibly improved estimator $\hat{ \bmath{B}}^*$ of ${\bm B}_0$ can be obtained. We call this efficiency improved estimation method iMAVE2.

Other approaches to obtain $\eta^*$ can also be considered. For example, it may be estimated from an external independent dataset or given directly through prior knowledge. 	When $\eta^*$  can not be estimated reliably, especially when the dimensionality of $\bmath{X}$ is high or when the sample size $n$ is small, as long as the estimator is a function of ${\bm X}$, the resulting $\hat{\bm B}^*$ is still unbiased in principle. Therefore instead of nonparametric estimators, parametric models may also be used to estimate $\eta^*$.



\subsection{Dimension determination}
\label{subsect:dimension determine}

There is a need to determine the dimension $d$, especially when $p$ is large. Many methods proposed in the dimension reduction literature are applicable in our setting too \citep{Koch2007, Schott1994, Cook1998}. In this paper, we adopt the same procedure as \citet{Xia2002}, which is a consistent procedure based on cross-validation. {In particular, because
	\begin{equation*}
	E\left[\frac{T}{\pi_T(\mathbf{X})}Y\middle \vert \mathbf{X}\right]=E\left [\frac{T}{\pi_T(\mathbf{X})}Y\middle \vert\mathbf{B}_0^\top \mathbf{X}\right ],
	\end{equation*}
	consistency of the dimension determination procedure can be established by a direct application of Theorem 2 in \citet{Xia2002}.}

Given a dimension $d\in \{0, 1, \cdots, p\}$, the procedure goes through the following steps based in iMAVE.
\begin{enumerate}
	\item Randomly split the dataset into five folds, and $\mathcal{I}_m, m=1,\cdots, 5$ are the index sets corresponding to these folds.
	\item For $m=1,\cdots, 5$, choose $\mathcal{I}_m$ as a testing set and the rest $\mathcal{I}_{-m}$ as a training data set. Fit iMAVE on $\mathcal{I}_{-m}$ to obtain estimates of $\hat{\bmath{B}}_{(-m)}$ and $\hat{g}_{(-m)}(\cdot)$. Then calculate the following score.
	\begin{equation*}
	CV(d, m)=\frac{1}{|\mathcal{I}_m|}\sum_{i\in \mathcal{I}_m}\left(\frac{1}{2}\frac{T_iY_i}{\pi_{T_i}(\bmath{X}_i)}-\hat{g}_{(-m)}(\hat{\bmath{B}}^\top \bmath{X}_i)\right)^2,
	\end{equation*}
	where $\hat{g}_{(-m)}(\cdot)$ is estimated using all other folds except the $m$th fold.
	\item  The estimated dimension is
	$
	\hat{d}=arg\min_{0\leq d\leq p} \sum_{m=1}^5 CV(d, m).
	$
\end{enumerate}

These same steps can be used based on iMAVE2 to determine the dimension. 
It is intuitively clear that over-estimating the true dimension $d$ to a slightly larger value is much less of a concern than under-estimating.


\section{Theoretical results}
\label{s:theory}

In this section, we analyze our estimator in a unified framework of statistical and algorithmic properties assuming a binary $T$ for notational simplicity. We study both iMAVE and iMAVE2. 

The non-convexity of (\ref{eq:tMAVE}) makes it intractable to obtain theoretical results for prediction or classification error by simply mimicking the usual analysis of empirical risk minimization \citep{Vapnikbook}. It is also hard to analyze the convergence rate or asymptotic distribution of the proposed estimators due to a lack of characterization of the minimizers. On the other hand, because we carry out our optimization by iteratively solving a weighted least square problem, we can track the change of each iteration similar to \citet{Xia2002} and \citet{Xia2007}. This leads us to propose a unified framework of joint statistical and algorithmic analysis.

For any matrix $\bmath{A}$, $|\bmath{A}|$ represents the Frobenius norm of $\bmath{A}$. For any random matrix $\bmath{A}_n$, we say $\bmath{A}_n=O_p(a_n)$ if each entry of $\bmath{A}_n$ is $O_p(a_n)$. Let $\bmath{B}_{(t)}$ be the estimator in the $t$th iteration of the iMAVE algorithm, and $\hat{\bmath{B}}$ be the limit of $\bmath{B}_{(t)}$ when $t\to +\infty$. The existence of the limit of $\bmath{B}_{(t)}$ as well as the convergence of the algorithm, similar to \cite{Xia2007}, can be concluded from the proof. Denote $\delta_{\bmath{B}}^{(t)}=|\bmath{B}_{(t)}-\bmath{B}_0|$. Our goal is to answer the following questions for both iMAVE and iMAVE2:
\begin{enumerate}
	\item Suppose that $\delta_{\bmath{B}}^{(1)}$ has some convergence rate to $0$. After $t$ iterations, what is the convergence rate of $\delta_{\bmath{B}}^{(t)}$? \label{Q1}
	\item What is the convergence rate of $\delta_{\hat{\bmath{B}}}\equiv |\hat{ \bmath{B}}-\bmath{B}_0|$? \label{Q2}
	
	\item What are the answers for Questions 1 and 2 when iMAVE2 is used.  \label{Q3}
	
	\item Whether there is asymptotic efficiency gain of iMAVE2 compared with iMAVE? \label{Q4}
\end{enumerate}
Questions \ref{Q1} and \ref{Q2} are answered by Theorems \ref{thm2} and \ref{thm1}, respectively. Question \ref{Q3} is answered by Theorem \ref{thm:efficient}. Question \ref{Q4} is answered by Theorems~\ref{thm3} and \ref{thm:efficient}. 

Theorem~\ref{thm2} is a new result beyond \citet{Xia2002} and \citet{Xia2007}. It essentially quantifies the non-asymptotic property of our estimators. It implies that under certain conditions, $\delta_{\bmath{B}}^{(t)}$ converges to $0$ with a rate at least $(n/\log n)^{-1/2}$ almost surely when $t$ is large enough and $d\leq 5$. When $d>5$, the convergence rate is bounded by a quantity related to bandwidth and $d$, and slower than $(n/\log n)^{-1/2}$. Theorem~\ref{thm1} implies that under certain conditions, $\delta_{\hat{\bmath{B}}}$ converges to $0$ in probability with the order of $n^{-1/2}$ when $d\leq 5$. When $d>5$, the convergence rate is slower than $n^{-1/2}$. The convergence rate in Theorem~\ref{thm1} is different than that in Theorem~\ref{thm2} by a factor of $\log n$ due to the difference of convergence modes. Theorem~\ref{thm2} provides deeper results with both statistical and algorithmic properties.

Theorems~\ref{thm3} and \ref{thm:efficient} provide the asymptotic distributions of iMAVE and iMAVE2 estimators, respectively. Theorem \ref{thm6} provides the accuracy of estimating $g$ based on $\hat{\bmath{B}}$. Combining with the previous results in Section \ref{s:method}, we will see that difference of the asymptotic covariance matrices of iMAVE and iMAVE2 is always positive semi-definite. Thus, iMAVE2 is more efficient than iMAVE.

The conditions needed for our theorems are as follows.  Let $\xi_{\bmath{B}}(\bmath{u})=E (\bmath{X}\bmath{X}^\top| \bmath{B}^\top \bmath{X}=\bmath{u})$ and $\mu_{\bmath{B}}(\bmath{u})\equiv E(\bmath{X}| \bmath{B}^\top  \bmath{X}=\bmath{u})$. We denote the distribution of $ \bmath{B}^\top \bmath{X}$ as $p_{\bmath{B}}( \bmath{B}^\top \bmath{x})$.
\begin{enumerate}
	\item[(C.1)] The density of $\bmath{X}$, $p_{\bmath{X}}(\bmath{x})$,  has bounded 4th order derivatives and compact support. $\mu_{ \bmath{B}}(\bmath{u})$ and $\xi_{ \bmath{B}}(\bmath{u})$ have bounded derivatives with respect to $  \bmath{u}$ and $ \bmath{B}$ where $ \bmath{B}$ is in a small neighborhood of $ \bmath{B}_0:| \bmath{B}- \bmath{B}_0|\leq \delta$, for some $\delta>0$.
	
	\item[(C.2)] The matrix $\bmath{M}_0=\int \nabla g(\bmath{B}_0^\top \bmath{x})\nabla^\top g(\bmath{B}_0^\top \bmath{x})\times p_{\bmath{B}_0}(\bmath{B}_0^\top \bmath{x})p_{\bmath{X}}(\bmath{x})d\bmath{x}$ has full rank $d$.
	
	\item[(C.3)] $K(\cdot)$ is a spherical symmetric univariate density function with a bounded 2nd order derivative and compact support.
	
	\item[(C.4)] $g$ has a bounded derivative. The error $\epsilon$ satisfies that there exists some $M$ and $\nu_0\in [0,+\infty)$ such that
	\begin{equation*}\label{cond:epsilon}
	E\left \{\exp \bigg [\frac{T\epsilon}{\pi_T({\bm X}) M} \bigg ]-1- \frac{|T\epsilon|}{\pi_T({\bm X}) M} \, \big| \, \bmath{X}\right \} M^2\leq \nu_0/2.
	\end{equation*}

	\item[(C.5)] The bandwidth $h_1=c_1n^{-r_h}$, where $0<r_h\leq 1/\{\max(p,3)+6\}$. For $t\geq 2$, $h_t=\max\{n^{-r_h/2}h_{t-1},\hbar\}$, where  $\hbar=c_3n^{-r'_h}$ with $0<r'_h\leq 1/(d+3)$. Here $c_1$ - $c_4$ are constants.
	
	\item[(C.6)] $p_{\bmath{B}}( \bmath{B}^\top \bm x)$ is bounded away from 0. In addition, $E[\pi_T({\bm X})^{-1}TY|\bmath{B}^\top \bmath{X}=\bmath{u}]$ is Lipschitz continuous {and $\pi_T({\bm X})$ is bounded away from 0 and 1}.
\end{enumerate}

Condition (C.6) is only needed for Theorem \ref{thm6}. Conditions (C.1) - (C.5) are similar to \citet{Xia2007} except the requirement for compact support of covariates. This requirement is needed for iMAVE2 because $g$ needs to be estimated to a certain rate for the asymptotic property of iMAVE2. For iMAVE, this requirement can be replaced by a finite moment condition. Epanechnikov and quadratic kernels satisfy Condition (C.3). The Gaussian kernel can also be used to guarantee our theoretical results with some modification to the proofs. According to \citet{Xia2007}, Condition (C.2) suggests that the dimension $d$ can not be further reduced. 
The bandwidth requirement in Condition (C.5) can be easily met. Condition (C.6) characterizes the smoothness of $g$ as typically required for conditional expectation estimation.


\begin{theorem}\label{thm2}
	Under Conditions (C.1) - (C.5), suppose that the initial estimator for iMAVE, $\bmath{B}_{(1)}$, satisfies $\delta_{\bmath{B}}^{(1)}/h_1\to 0$, if $n$ is large enough, then there exists a constant $C_1$ such that when the number of iterations $t$ satisfies $$t\geq 1+\log \min\left \{\frac{3C_1\{\delta_n+\delta_{d\hbar}^2\hbar+\hbar^4\}}{\delta_{\bmath{B}}^{(1)}+2C_1h_1^4},1\right \} \bigg /\log \frac{2}{3},$$ we have $
	\delta_{\bmath{B}}^{(t)}\leq (3C_1+1)\{\delta_n+\delta_{d\hbar}^2\hbar+\hbar^4\} \text{almost surely},
	$
	where $\delta_n=(n/\log n)^{-1/2}$ and $\delta_{d\hbar}=(n\hbar^d/\log n)^{-1/2}$.
\end{theorem}

A simple observation from Theorem \ref{thm2} implies that to reach the same accuracy when $d$ increases, the number of iterations required is increasing linearly in $d$. This provides a useful guidance on the maximum number of iterations for the algorithm.

\begin{theorem}\label{thm1}
	Under the same conditions as Theorem~\ref{thm2}, there exists a matrix $\bmath{B}_0^{\perp}$ whose column space is the orthogonal complement of the column space of $\bmath{B}_0$, such that the iMAVE estimator satisfies 
	$$
	\hat{\bmath{B}}=\bmath{B}_0\, \{\bm I_d+O_p(\hbar^4+\delta_{d\hbar}^2+n^{-1/2})\}+\bmath{B}_0^{\perp}\, O_p(\hbar^4+\delta_{d\hbar}^2+n^{-1/2}).
	$$
\end{theorem}

Theorem~\ref{thm1} implies that when $\hat{\bmath{B}}$ is decomposed based on the column space of $\bmath{B}_0$ and its orthogonal complement, the component in the column space of $\bmath{B}_0^{\perp}$ converges to $0$, and the projection of $\hat{\bmath{B}}$ on the column space of $\bmath{B}_0$ converges to $\bmath{B}_0$. To obtain the $n^{-1/2}$ convergence rate, we need $\hbar^4+\delta_{d\hbar}^2=O(n^{-1/2})$. In this case, $d$ has to be smaller than $5$.

\begin{theorem}\label{thm3}
	Assume the same conditions as Theorem \ref{thm2} and  $\hbar^4+\delta_{d\hbar}^2=o_p(n^{-1/2})$. Denote $\nu_{\bmath{B}}(\bmath{x}) \equiv \mu_{\bmath{B}}(\bmath{B}^\top  \bmath{x})-\bmath{x}$. Let $l(\hat{\bmath{B}})$ and $l(\bmath{B}_0)$ be vectorizations of the matrices $\hat{\bmath{B}}$ and $\bmath{B}_0$, respectively. Then
	\begin{equation*}
	\sqrt{n}\{l(\hat{\bmath{B}})-l(\bmath{B}_0)\}\to N(0,\bmath{D}_0^+\bmath{\Sigma}_0\bmath{D}_0^+),
	\end{equation*}
	where $\bmath{\Sigma}_0=Var\left [\pi_{T_i}(\bmath{X}_i)^{-1} {T_i}\nabla g(\bmath{B}_0^\top  \bmath{X}_i)\otimes \nu_{\bmath{B}_0}(\bmath{X}_i)\{\epsilon_i-\eta(\bmath{X}_i)\}\right ]$. The expression of $\bmath{D}_0^+$ can be found in our proof of this theorem from the Supplemental Materials.
\end{theorem}

\begin{theorem} \label{thm6}
	Suppose that Conditions (C.1) - (C.6) are satisfied and $g$ is estimated by kernel $K_{h}$ of some order $m$. Then $h$ can be selected such that when $n$ is large enough,
	\begin{equation*}
	\|\hat{ g}(\hat{\bm B}^\top \bm X)-g(\bm B_0^\top\bm X)\|_{\infty}\leq O\Big\{ {(n/\log n)}^{-\frac{m}{2m+d}} \Big \}, \text{almost surely}.
	\end{equation*} where $m$ can be any integer when $d\leq 5$, but  $m \le {4d}/{(d-5)}$ when $d> 5$
\end{theorem}

\begin{theorem}\label{thm:efficient}
	Denote $\delta_{ph}\equiv (nh^p/\log n)^{-1/2}$.  In iMAVE2, suppose $d\leq 5$ and $\delta_{ph}^2+h^{2m}=o(n^{-1/2})$ when estimating $\eta^*$ by $\hat{\eta}^*$ using \eqref{eq:estimation_mainEffect2} or \eqref{eq:estimation_mainEffect}. Then, under Conditions (C.1) - (C.5), for iMAVE2, Theorems \ref{thm2} and \ref{thm1} still hold and Theorem \ref{thm3} holds with the asymptotic variance, $\bmath{D}_0^+\bmath{\Sigma}_0^*\bmath{D}_0^+$, where $\bmath{\Sigma}_0^*=Var\Big[\pi_{T_i}(\bmath{X}_i)^{-1} {T_i}\nabla g(\bmath{B}_0^\top  \bmath{X}_i)\otimes \nu_{\bmath{B}_0}(\bmath{X}_i)\{\epsilon_i-\eta^*(\bmath{X}_i)\}\Big]$, and $\bmath{\Sigma}_0-\bmath{\Sigma}_0^*$ is positive semi-definite.
\end{theorem}


Detailed proofs for all theorems are given in the Supplemental Materials.  Here we consider construction of confidence intervals for $\mathbf{B}_0$ and possible improvement of empirical estimation with limited sample sizes. From Theorems~\ref{thm3} and~\ref{thm:efficient}, we know that the estimators are both $\sqrt{n}$-consistent and asymptotically normal under suitable conditions. This makes the inference of $\mathbf{B}_0$ possible if we have a stable way to estimate the asymptotic variances to form confidence intervals. In theory we just need to evaluate the variance formulae using observed data.

However, we found from our simulation studies that estimation of $\nabla g$ in the asymptotic variance formulae can be challenging. If we directly use all the data to estimate $\nabla g$, the resulting confidence intervals often over cover. This is because estimation of $\nabla g$ is directly related to estimation of $\mathbf{B}_0$. Using data twice to first estimate $\mathbf{B}_0$ and then estimate $\nabla g$  leads to overfitting. Therefore, we propose a sample split procedure to alleviate this issue, similar to some recent works \cite{Chernozhukov2018, Yingqi2019, athey2017efficient}.  Specifically, the whole data set is split into halves randomly. On the first half, an iMAVE or iMAVE2 estimate of $\mathbf{B}_0$ is obtained. On the other half, we estimate $g$ and $\nabla g$ using smoothing splines.

In addition, we found that a further one-step Newton-Raphson estimator for $\mathbf{B}_0$ can lead to some improvement, especially when the sample size is limited. In particular we use the following step:
\begin{equation*}
\hat{\mathbf{B}}_{NR}=\hat{\mathbf{B}}_{MV}-\left\{E^{(1)}\left[\frac{\partial \hat{\mathcal{S}}(\hat{\mathbf{B}}_{MV}; \bmath{X}, T, Y)}{\partial \hat{\mathbf{B}}_{MV}} \right]\right\}^{-1}E^{(1)}\left[\hat{\mathcal{S}}(\hat{\mathbf{B}}_{MV}; \bmath{X}, T, Y)\right],
\end{equation*}
where
\begin{equation*}
\hat{\mathcal{S}}(\hat{\mathbf{B}}_{MV}; \bmath{X}, T, Y) = \pi_{T}(\bmath{X})^{-1} {T}\nabla \hat{g}(\hat{\mathbf{B}}_{MV}^\top \bmath{X})\otimes \hat{\nu}_{\hat{\mathbf{B}}_{MV}}(\bmath{X}_i) (\epsilon-\eta(\bmath{X})),
\end{equation*}
$\hat{\mathbf{B}}_{MV}$ is the iMAVE or iMAVE2 estimator with corresponding choice of $\eta$ or $\hat{\eta}^*(\mathbf{X})$ on the first half of the dataset, $\hat{\nu}_{\hat{\mathbf{B}}_{MV}}$ is the estimator of $\nu_{\bmath{B}_0}$ on the second half of the dataset, and $\nabla\hat{g}$ is the estimator of gradient on the second half of the dataset. $E^{(1)}[\cdot]$ represents expectation taken over the first half of the dataset. From the theory of one-step Newton-Raphson estimators, $\hat{\mathbf{B}}_{NR}$ is still a $\sqrt{n}$-consistent estimator and its asymptotic variance can be estimated by
\begin{equation*}
\left\{E^{(1)}\left[\frac{\partial \hat{\mathcal{S}}(\hat{\mathbf{B}}_{MV}; \bmath{X}, T, Y)}{\partial \hat{\mathbf{B}}_{MV}} \right]\right\}^{-1}Var\left[\hat{\mathcal{S}}(\hat{\mathbf{B}}_{MV}; \bmath{X}, T, Y)\right]\left\{E^{(1)}\left[\frac{\partial \hat{\mathcal{S}}(\hat{\mathbf{B}}_{MV}; \bmath{X}, T, Y)}{\partial \hat{\mathbf{B}}_{MV}} \right]\right\}^{-1}.
\end{equation*}
Due to the sample split procedure, the estimation error of $\hat{g}$ is not related to the first half of the data, which results in a more stable estimation of the asymptotic variance.

\section{Simulation}
\label{s:simulation}

Here our method is evaluated and compared with existing methods. In particular, we compare with the outcome weighted learning method based on a logistic loss in \citet{Xu2015}, the modified covariate method under the squared loss proposed in \citet{Tian2014}, and residual weighted learning method \citep{Zhou2017} based on a logistic loss. We also compare with Q-learning with linear basis functions as a parametric version of the proposed loss function \citep{qian2011}. We first evaluate estimation results assuming $d$ is known and then investigate dimension determination. {Given the fact that  \citet{song2017} is a special case of iMAVE and their method can be applied only when $d=1$, we do not include it as our comparison method.}

We report part of the results for estimating effect modification and dimension determination in the main text. The rest of the simulation results are relegated to the Supplemental Materials.   There we also report confidence interval coverage, results for additional settings including  more complex data generation models and correlated covariates. 

\subsection{Estimation evaluation with known $d$}
\label{subsect:simulation_estimation}

Data are generated by the following model,
\begin{eqnarray}
y=({\bm \beta}^\top  \bm X)^2+\frac{1}{2}Tg({\bm \beta}^\top  \bm X)+\epsilon,
\label{simulation model}
\end{eqnarray}
where $\epsilon\sim N(0,\sigma^2)$ and $g$ is chosen as
\begin{enumerate}
	\item Linear: $g({\bm {\bm \beta}}^\top  \bm X)=\tau{\bm \beta}^\top  \bm X$;
	\item Logistic: $g({\bm \beta}^\top  \bm X)=\tau\{ (1+e^{-{\bm \beta}^\top  \bm X})^{-1}-0.5\}$;
	\item Gaussian: $g({\bm \beta}^\top  \bm X)=\tau\{\Phi({\bm \beta}^\top  \bm X)-0.5\}$, where $\Phi(\cdot)$ is the Gaussian distribution function.
\end{enumerate}

We set $\sigma=0.6$, $\tau=7$, and $T$ is generated to be $-1$ or $1$ with equal probability and independent with all other variables. The true ${\bm \beta}_0$ is chosen to be $(1,1,1,1)^\top $. $\bm X$ is generated from $N(0,\bm I_{4\times4})$. The sample size $n$ varies from $200$, $500$ to $1000$. Results are summarized from 1000 simulated data sets.

Table~\ref{tab:coef_low} investigates the asymptotic bias of the iMAVE and iMAVE2 and the possible gain in efficiency from the latter. The ratios $\hat{\beta}_j/\hat{\beta}_1$, $j=2,3,4$, are reported due to the Grassmann manifold assumption for identifiability. Whereas there are some empirical biases for nonlinear $g$ under small sample sizes, as the sample size increases, the means of the ratios all approach $1$, the true value. There is noticeable improvement from iMAVE2 over iMAVE in terms of MSE. 

We further consider prediction results under the settings of known and estimated propensity scores. In particular we investigate the estimated effect modification in terms of correct classification rate and rank correlation over test data sets generated independently according to the true simulation model above but with sample sizes of $10000$. The rank correlation is determined by 
the fitted classifier and the true $g({\bm \beta}_0^\top  \bm X)$ and the classification rate by their corresponding signs.  For example, for iMAVE and iMAVE2,  we evaluate the rank correlation between $\hat{g}(\hat{\bm \beta}^\top\! \bm X)$ and $g({\bm \beta}_0^\top  \bm X)$ and the concordance between $\hat{g}(\hat{\bm \beta}^\top\! \bm X)>0$ and $g({\bm \beta}_0^\top  \bm X)>0$ to determine the correct classification rate.

{In our simulation setting where $g$ is monotone and $g(0)=0$, the sign of $g({\bm \beta}_0^\top  \bm X)$ is also identical to that of ${\bm \beta}_0^\top  \bm X$. In addition, the rank correaltion between $g(\hat{\bm \beta}^\top\! \bm X)$ and $g({\bm \beta}_0^\top  \bm X)$ is also identical to that between $\hat{\bm \beta}^\top\! \bm X$ and ${\bm \beta}_0^\top  \bm X$. Because the resulting estimators of \cite{Tian2014},  \cite{Xu2015}, and \cite{Zhou2017}  are parametric and target at the decision boundary ${\bm \beta}_0^\top  \bm X$, we also include results of iMAVE(index) and iMAVE2(index) which compare the concordance between $\hat{\bm \beta}^\top\! \bm X>0$ and ${\bm \beta}_0^\top  \bm X>0$ and the rank correlation between $\hat{\bm \beta}^\top\! \bm X$ and ${\bm \beta}_0^\top  \bm X$ when $g$ is monotone and $g(0)=0$. This represents a more fair comparison with the parametric methods. Again, the index comparison only makes sense when $g$ is monotone which is the case in our simulation setting. }

From Figure~\ref{fig:plot1}, our methods have the best correct classification rates for the test datasets in all settings with known propensity score. 
When $g$ is monotone and $g(0)=0$, in terms of rank correlation, iMAVE2(index) is the best followed by iMAVE(index). The performances of iMAVE and iMAVE2 sacrifice slightly due to the estimation of $g$. 

{We further investigate the setting when $\pi_T({\bm X})$ needs to be estimated. In this case, we generate $T$ from a logistic model with coefficients $\tilde{\bm \beta}=(0.2,-0.2,0.2,-0.2)^\top$ and then fit a logistic regression  for $\pi_T({\bm X})$. After estimating $\pi_T({\bm X})$, all methods are implemented with the estimated $\pi_T({\bm X})$.}
From Figure~\ref{fig:plot2}, our methods have the best correct classification rate and rank correlation than all other methods in all settings.

\begin{table}
	\caption{Simulation results for coefficient estimation.}
	\label{tab:coef_low}
	\begin{center}
		\begin{tabular}{ccccccccc}
			\toprule
			\multirow{2}{*}{Size}& & & \multicolumn{2}{c}{Linear} & \multicolumn{2}{c}{Gaussian} & \multicolumn{2}{c}{Logistic}  \\[-0.5ex]
			\cmidrule(lr){4-5} \cmidrule(lr){6-7} \cmidrule(lr){8-9}  
			& & & iMAVE &  iMAVE2 & iMAVE &  iMAVE2 & iMAVE &  iMAVE2\\[-.25ex]
			\hline\\[-2.5ex]
			\multirow{6}{*}{200}& \multirow{3}{*}{mean} & $\hat{\beta}_2/\hat{\beta}_1$ & 0.9995 & 0.9986 & 0.8630 & 0.9161 & 0.7797 & 0.8611 \\
			& & $\hat{\beta}_3/\hat{\beta}_1$ & 1.0021 & 1.0021 & 0.8960&0.9410 & 0.8192 & 0.8884 \\
			& & $\hat{\beta}_4/\hat{\beta}_1$ & 1.0042 & 1.0035 & 0.8891&0.9408 & 0.8013 & 0.8802 \\[1ex]
			\cline{3-9} \\[-2ex]
			& \multirow{3}{*}{$\sqrt{mse}$} & $\hat{\beta}_2/\hat{\beta}_1$ & 0.0563 & 0.0378 & 0.3122&0.2044  & 0.4106 & 0.2890 \\
			& & $\hat{\beta}_3/\hat{\beta}_1$ & 0.0586 & 0.0386 &0.2971&0.1977 & 0.4056 & 0.2837\\
			& & $\hat{\beta}_4/\hat{\beta}_1$ & 0.0540 & 0.0361 & 0.3075&0.2055 & 0.4191 &0.2847\\[0.25ex]
			\hline\\[-2.5ex]
			\multirow{6}{*}{500}& \multirow{3}{*}{mean} & $\hat{\beta}_2/\hat{\beta}_1$ & 0.9978 &0.9994 & 0.9526&0.9759& 0.8995&0.9484\\
			& & $\hat{\beta}_3/\hat{\beta}_1$ & 1.0010 & 1.0004 & 0.9701&0.9854 & 0.9193& 0.9625 \\
			& & $\hat{\beta}_4/\hat{\beta}_1$ & 1.0020 & 1.0004 & 0.9452 & 0.9798 & 0.8994& 0.9477\\[1ex]
			\cline{3-9} \\[-2ex]
			& \multirow{3}{*}{$\sqrt{mse}$} & $\hat{\beta}_2/\hat{\beta}_1$ & 0.0372 & 0.0207 & 0.1676&0.0975 & 0.2539& 0.1558 \\
			& & $\hat{\beta}_3/\hat{\beta}_1$ & 0.0329 & 0.0188 & 0.1663&0.0935& 0.2587 & 0.1507 \\
			& & $\hat{\beta}_4/\hat{\beta}_1$ & 0.0326 & 0.0184& 0.1675&0.0925& 0.2531 & 0.1505 \\[0.5ex]
			\hline\\[-2.5ex]
			\multirow{6}{*}{1000}& \multirow{3}{*}{mean} & $\hat{\beta}_2/\hat{\beta}_1$ & 1.0015 & 1.0006 & 0.9994&1.0032 & 0.9728&0.9913 \\
			& & $\hat{\beta}_3/\hat{\beta}_1$ & 1.0009 & 1.0007 & 1.0020&1.0026 & 0.9794&0.9946\\
			& & $\hat{\beta}_4/\hat{\beta}_1$ & 0.9993 & 1.0006 & 0.9980& 1.0018 & 0.9756 & 0.9897 \\[1ex]
			\cline{3-9} \\[-2ex]
			& \multirow{3}{*}{$\sqrt{mse}$} & $\hat{\beta}_2/\hat{\beta}_1$ & 0.0233 & 0.0124 & 0.1014&0.0515 & 0.1656&0.0905 \\
			& & $\hat{\beta}_3/\hat{\beta}_1$ & 0.0247 & 0.0125 & 0.1017&0.0533& 0.1672&0.0894 \\
			& & $\hat{\beta}_4/\hat{\beta}_1$ & 0.0236 & 0.0123 & 0.1033&0.0520 & 0.1627 & 0.0885 \\
			\bottomrule
		\end{tabular}
	\end{center}
\end{table}

\begin{figure}
	\vspace{6pc}
	\includegraphics[scale=0.65]{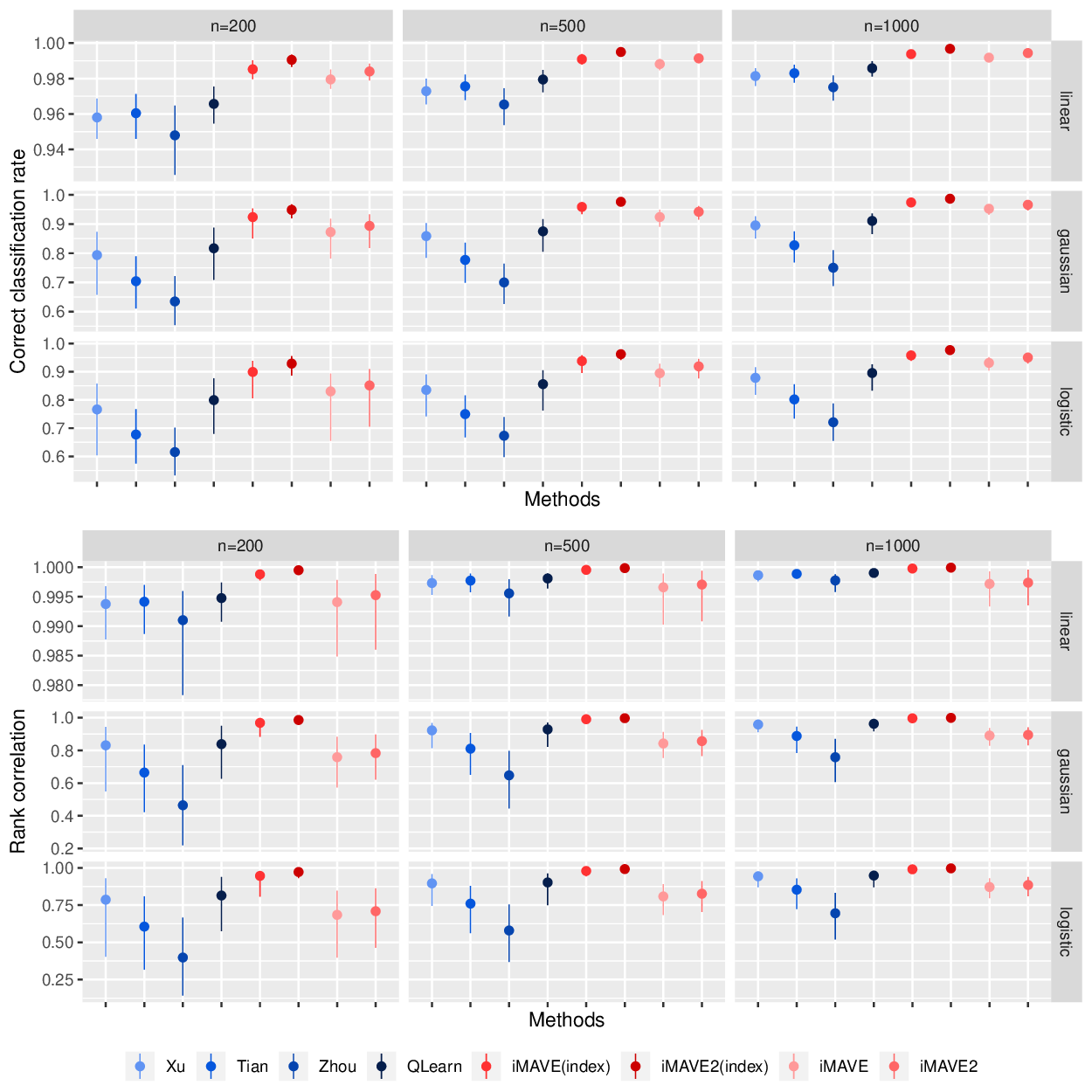}
	\caption{Simulation results for rank correlation and classification rate with {\bf known}  $\pi_T({\bm X})$. The point represents the median, and the vertical line represents the range from the $0.25$ to the $0.75$ quantiles, of the results from 1000 simulations.
		\label{fig:plot1}}
\end{figure}

\begin{figure}
	\vspace{6pc}
	\includegraphics[scale=0.65]{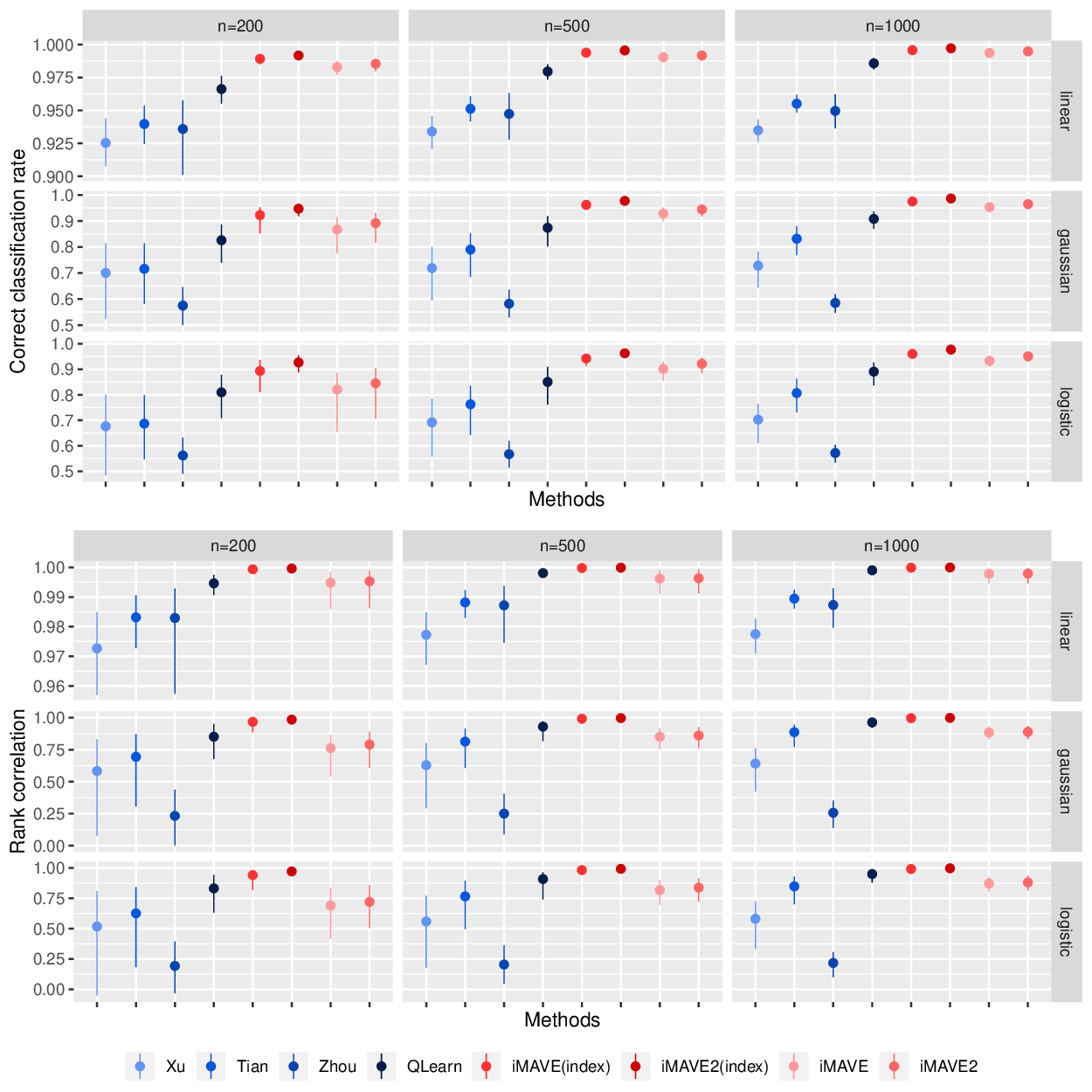}
	\caption{Simulation results for rank correlation and classification rate with {\bf estimated}  $\pi_T({\bm X})$. The point represents the median, and the vertical line represents the range from the $0.25$ to the $0.75$ quantiles, of the results from 1000 simulations.
		\label{fig:plot2}}
\end{figure}

\subsection{Dimension determination}

Here we evaluate our dimension determination procedure through simulation. We follow  Section \ref{subsect:simulation_estimation} mostly except that we set $p=10$ and the true $d=2$. Consequently, the function $g$ is 
\begin{equation*}
g(\bm B^\top \bm X)=\tau\{\Phi(\bm \beta_1^\top  \bm X)-0.5\}+\tau\{\Phi(\bm \beta_2^\top  \bm X)-0.5\}
\end{equation*}
where $\bm \beta_1=(1,1,1,1,1,1,1,1,1,1)^\top$ and $\bm \beta_2=(1,-1,1,-1,1,-1,1,-1,1,-1)^\top.$ We set $\gamma=0.1$ and the sample size $n$ is fixed at $500$. Over $100$ simulated data sets, our procedure was able to choose the correct dimension $2$ for $72$ times, $3$ for $26$ times, and $4$ for $2$ time. As we mentioned before, over-estimating the dimension slightly is not a big issue. There is no under-estimation of $d$, but slight over-estimation in some data sets.

\section{Application to a mammography screening study}

This is a randomized study that included female subjects who were non-adherent to mammography screening guidelines at baseline (i.e., no mammogram in the year prior to baseline) \citep{champion2007effect}. One primary interest of the study was to compare the intervention effect of \textcolor{black}{phone counseling on mammography screening (phone intervention)} versus usual care at 21 months post-baseline. The outcome is whether a subject took mammography screening during this time period. There are 530 subjects with 259 in the phone intervention group and 271 in the usual care group. Baseline covariates include socio-demographics, health belief variables, stage of readiness to undertake mammography screening, and number of years had a mammogram in past 2 to 5 years in the study.  In total, there are $211$ covariates including second order interactions among the covariates. 

Our methods, together with our comparator methods \cite{Xu2015,Tian2014, Zhou2017}, were applied to this data set. To compare the results of the estimated treatment assignment rules, we used the following metrics.  An assignment rule $T({\bm X})$ refers to a mapping from ${\bm X}$ to $\{1, -1\}$. For example, in our model set up with $\Delta({\bm X}) = E[Y|T=1, {\bm X}] - E[Y|T=-1, {\bm X}] =g({\bm \beta}^\top{\bm X}),$ the assignment rule that maximizes the expected value of the outcome is $T({\bm X})=1\{g({\bm \beta}^\top{\bm X})>0\}$. For a fitted assignment rule, say $\hat{T}(\bm X)$, the following two quantities are used to evaluate the performances.
\begin{equation*}
E[\Delta_1]=E[Y|\hat{T}(\bm X) = 1, T = 1]-E[Y|\hat{T}(\bm X) = 1, T = -1],
\end{equation*}
and,
\begin{equation*}
E[\Delta_{-1}]=E[Y|\hat{T}(\bm X) = -1, T = -1]-E[Y|\hat{T}(\bm X) = -1, T = 1].
\end{equation*}
They represent gains in the outcome expectations between the recommendation agreeing and disagreeing subgroups. If both $E[\Delta_{-1}]$ and $E[\Delta_1]$ are positive, then the estimated treatment decision rule can improve the outcome. 

The actual evaluation was based on cross-validation. First, $80\%$ of subjects were randomly selected into a training set and the rest into a testing set. 
Apparently, due to this further reduction of sample size, we had to reduce the number of covariates for fitting. We performed screening procedures for all methods in a uniform fashion. In particular, the method of  \cite{Tian2014} with lasso penalty was fitted on the training sets for variable selection. After variable selection, the selected covariates were fitted by each method. For iMAVE and iMAVE2, dimension selection from $d=1,2,3$ was also implemented. Then, the benefit quantities defined above were calculated on the testing set. The cross-validation was based on 100 splits. The SDs in Table \ref{tab:compare_realdata} refer to the standard deviations of 
$\hat{E}[\Delta_1]$ and $\hat{E}[\Delta_{-1}]$ from these 100 repeats. In Table \ref{tab:compare_realdata}, our methods seem to have advantages as they lead to larger $\hat{E}[\Delta_1]$ and $\hat{E}[\Delta_{-1}]$. The average percentages of subjects assigned to $T=1$ and $-1$ in the test sets are also given in the table. {A list of the top selected variables by the screening method is provided in the Supplementary Materials.}

\begin{table}
	\caption{Results for the mammography screening study from 100 cross validations. }
	\label{tab:compare_realdata}
	\begin{center}
		\begin{tabular}{ccccc}
			\toprule
			& \multicolumn{2}{c}{$\hat{E}[\Delta_1]$} & \multicolumn{2}{c}{$\hat{E}[\Delta_{-1}]$}\\
			\cmidrule(lr){2-3} \cmidrule(lr){4-5} \\[-1.75ex]
			&  & Avg \% of &  & Avg \% of  \\
			Method & Mean (SD)&  subj in $T=1$ & Mean (SD)& subj to $T=-1$ \\
			\hline  \\[-1.5ex]
			iMAVE & 0.032(0.014) & 42\% & 0.052(0.012) & 58\% \\
			iMAVE2 & 0.036(0.014) & 42\% & 0.054(0.012) & 58\% \\
			Tian & 0.022(0.013) & 44\% & 0.043(0.011) & 56\%\\
			Xu & 0.026(0.012) & 43\% & 0.044(0.012) & 57\%\\
			Zhou& 0.020(0.013) &41\% &0.041(0.011) &59\% \\
			QLearn & 0.018(0.012)& 33\%&0.022(0.011)&67\%\\
			\bottomrule
		\end{tabular}
	\end{center}
\end{table}

%

\section{Discussion}

In this article, we have proposed a very general semiparametric modeling framework for effect modification estimation. Whereas our main motivational setting is from precision medicine, the framework is generally applicable to statistical interaction discovery with interested variables in many other settings. For example in health disparities research, a complex and interrelated set of individual, provider, health system, societal, and environmental factors contribute to disparities in health and health care. Federal efforts to reduce disparities often include a focus on designated priority populations who are particularly vulnerable to health and health care disparities. Our approach seems ideal for data analysis in this setting. 

When there are many covariates, we have focused on dimension reduction. {In high dimensional settings, variable screening may be needed to reduce the number of covariates. Various methods can be applied in our framework. For example, because $E[TY/\pi_T|\bm X]=g(\bm\beta^\top\bm X)$, we can implement a non-parametric variable screening method such as the distance correlation based approach \cite{runze2012}. Alternatively, regression with penalty for variable selection such as lasso can be used \cite{Tian2014, Xu2015}. } Ideally, one could also incorporate variable selection into our framework when the dimension $d$ is fixed. In particular, lasso type of regularization can be used together with our estimating equations. This can be a fruitful path for future work as variable selection is an important practical issue.

\newpage
\bigskip
\begin{center}
{\large\bf Supplemental Materials}
\end{center}

\begin{description}

\item[] \hspace{.65cm} Estimation with multiple level treatments or exposures, proofs of Theorems \ref{thm:nuissance_multi}-\ref{thm:efficient}, additional simulation results, and supplemental results for the mammography screening study are contained in the Supplemental Materials\label{suppA}. 

\end{description}

\bibliography{ref.bib}
\end{document}